\documentclass[11pt]{article}

\usepackage[margin = 1 in]{geometry}
\usepackage{graphicx}
\usepackage{amstext,amsmath,amssymb,amsthm} 
\usepackage{hyperref}
\usepackage{authblk}

\newcommand{\basisname}{\mathcal{B}}
\newcommand{\vecl}{{\vec{l}}}
\newcommand{\vecb}{\vec{b}}
\newcommand{\Us}{U_S}
\newcommand{\Uc}{U_C}
\newcommand{\tUs}{\tilde{U}_S}
\newcommand{\tUc}{\tilde{U}_C}
\newcommand{\Z}{\mathbb{Z}}
\theoremstyle{plain}
\newtheorem{theorem}{Theorem}[section]

\newtheorem{lemma}[theorem]{Lemma}
\newtheorem{claim}[theorem]{Claim}
\newtheorem{corollary}[theorem]{Corollary}
\newtheorem{definition}[theorem]{Definition}
\newtheorem{algorithm}{Algorithm}
\newcommand{\prob}{{\rm Prob}}
\newcommand{\var}{{\rm Var}}
\newcommand{\cov}{{\rm Cov}}
\newcommand{\Exp}{{\rm E}}
\DeclareMathOperator{\Span}{span}

\title{How hard is deciding trivial versus nontrivial in the dihedral
  coset problem? }
\author{Nai-Hui Chia\thanks{Partially supported by National Science
    Foundation award CCF-1218721, and by the National Security Agency
    (NSA) under Army Research Office (ARO) contract number
    W911NF-12-1-0522.}} 
\affil{Department of Computer Science and Engineering, \\
    The Pennsylvania State University, University Park, PA 16802\\ 
\texttt{nxc233@cse.psu.edu}} 
\author{Sean Hallgren\thanks{Partially supported by National
    Science Foundation awards CCF-1218721  and CCF-1618287, and by the
    National Security Agency (NSA) under Army Research Office (ARO)
    contract number W911NF-12-1-0522.}}
\affil{Department of Computer Science and Engineering, \\
The Pennsylvania State University University Park, PA 16802\\ 
\texttt{hallgren@cse.psu.edu}}

\begin{document}

\maketitle

\begin{abstract}
  We study the hardness of the dihedral hidden subgroup problem.  It
  is known that lattice problems reduce to it, and that it reduces to
  random subset sum with density $> 1$ and also to quantum sampling
  subset sum solutions.  We examine a decision version of the problem
  where the question asks whether the hidden subgroup is trivial or
  order two.  The decision problem essentially asks if a given vector
  is in the span of all coset states.  We approach this by first
  computing an explicit basis for the coset space and the
  perpendicular space.  We then look at the consequences of having
  efficient unitaries that use this basis.  We show that if a unitary
  maps the basis to the standard basis in any way, then that unitary
  can be used to solve random subset sum with constant density $>1$.
  We also show that if a unitary can exactly decide membership in the
  coset subspace, then the collision problem for subset sum can be
  solved for density $>1$ but approaching $1$ as the problem size
  increases.  This strengthens the previous hardness result that
  implementing the optimal POVM in a specific way is as hard as
  quantum sampling subset sum solutions.
\end{abstract}

\section{Introduction}\label{sec:intro}

The dihedral coset problem is an important open problem in quantum
algorithms.  It comes from the hidden subgroup problem, which is
defined as: given a function on a group $G$ that is constant and
distinct on cosets a subgroup $H$, find $H$.  Here we will focus on
the case when $G$ is the dihedral group of order $2N$.  It is known
that this problem reduces to the case when the subgroup is order
two~\cite{Ettinger00}.  All known approaches for solving the hidden
subgroup problem over the dihedral group start by evaluating the
function in superposition and measuring the function value.  The
result is a random coset state $\tfrac{1}{\sqrt{2}}(|0,x\rangle +
|1,x+d\rangle)$, where $d\in \Z_N$ is a fixed label of the subgroup
and $x$ is a coset representative uniformly chosen in $\Z_N$.  For our
purposes, it is more convenient to have the following quantum problem
rather than the hidden subgroup problem.

The {\em dihedral coset problem}~\cite{Regev04} is: given a tensor product of $k$
coset states
\[
	|c_{x_1,x_2,\cdots,x_k}^{(d)}\rangle=\frac{1}{\sqrt{2}}(|0,x_1\rangle+|1,x_1+d\rangle)
        \otimes
        \cdots\otimes\frac{1}{\sqrt{2}}(|0,x_k\rangle+|1,x_k+d\rangle),
\]
where $x_1,\dots,x_k$ are randomly chosen in $\Z_N$, compute $d$.  The
first register of each state is mod 2, and the second register is mod
$N$.

This is a natural problem to consider after the successes with
abelian groups such as $\Z_N$.  The dihedral group with $2N$ elements
has $\Z_N$ as a normal subgroup.  The representations are mostly
two dimensional, so it does not have obvious problems like the
symmetric group, where we know large entangled measurements are
required to get information from the states~\cite{HMRRS10}.  Furthermore,
Regev~\cite{Regev04} showed that the unique shortest vector problem reduces
to the dihedral coset problem, so it could provide a pathway for finding a quantum
algorithm for lattice problems.

Much is known about the dihedral coset problem, at least compared to
most other nonabelian groups (although there are groups with efficient
algorithms, e.g.\ \cite{FIMSS14,Ivanyos2008,Decker2014}).  Ettinger and
Hoyer~\cite{Ettinger00} showed that a polynomial number of
measurements in the Fourier basis has enough classical information to
determine $d$, but the best known algorithm takes exponential time to
compute it.  Kuperberg found subexponential time
algorithms~\cite{Kuperberg05,Kup13} for the problem.  He also showed
that computing one bit of $d$ was sufficient to compute all of $d$.
This algorithm was a big step, although it should be noted that it
seems difficult to combine this with Regev's uSVP to dihedral group
HSP reduction to get a subexponential time algorithm for the uSVP, partly
due to the fact that the coset states created in the reduction have
errors with some probability.

The dihedral coset problem also has some connections to the subset sum
problem.  Bacon, Childs, and van Dam analyzed how well a ``pretty good
measurement'' performs~\cite{Bacon06}.  This type of measurement
maximizes the probability of computing $d$ correctly.  It is unknown
how to compute the measurement they find without quantum sampling
subset sum solutions.  A unitary implementing this can be used to
solve the worst case subset sum, which is NP-complete.  Regev showed how
to reduce the dihedral coset problem to the random subset sum problem
density $\rho>1$ where $\rho$ also approaches 1 as the problem size
increases.  Density $1$ is the hardest case for the random subset sum
problem as shown in Proposition 1.2 in~\cite{Impagliazzo96}.  But is
solving the dihedral coset problem as hard as subset sum, and if so,
for what parameters?  The only connection we are aware of is to
compose two known reductions.  First, random subset sum with density
$\rho=1/\log k$ reduces to uSVP.  Then uSVP reduces to the dihedral
coset problem.  It is open if an efficient quantum algorithm exists
for random subset sum, and density $1/\log k$ may not be as hard to
solve as constant density.

\subsection{New approach}

In this paper we focus on distinguishing trivial from order two
subgroups.  Instead of trying to compute $d$, we define a problem
which asks if the state is an order two coset state, or is the trivial
subgroup case.  We define this problem as the {\em dihedral coset
  space problem (DCSP)}: either an order two coset state is given, or
a random standard basis vector is given, decide which.  The random
standard basis vector corresponds to the trivial subgroup case in the
hidden subgroup problem.  This problem is a special case of the
decision version of the HSP defined by Fenner and Zhou~\cite{FZ08}
since we are restricting to order two subgroups.  In their paper,
they found a search to decision reduction when $N$ is a power of
two.  So it turns out that the problem is not computationally easier
in that case.

We start by finding a set of vectors that span $C$ and $C^\bot$.  Let
$\vecl \in \Z_N^k$, and $p\in \Z_N$.  The vectors
have the form 
\[
|S_{\vecl,p}^m\rangle = \tfrac{1}{\sqrt{|T_{\vecl,p}|}}
\sum\limits_{j=0}^{|T_{\vecl,p}|-1} \omega_{|T_{\vecl,p}|}^{mj}|
\vec{b}^{(j)}_{\vecl,p}\rangle|\chi_{\vecl}\rangle,
\]
where $T_{\vecl,p}$ contains the subset sum solutions for $(\vecl,p)$,
and the vectors $\vec{b}$ are an ordered set of the subset sum
solutions.  We call this set of orthonormal vectors the {\em subset
  sum} basis.  We prove that the $m=0$ subset of vectors span $C$ and
the remaining ones, which have $m\geq 1$, span $C^\bot$.

Ideally we would like to reduce subset sum to the DCSP.  Since this is
still out of reach, we prove a weaker relationship.  Instead, we
assume there is an algorithm that uses the subset sum basis to solve
the DCSP and examine the consequences.  Such an algorithm needs to
decide if $m=0$ or $m\geq 1$ to distinguish if the vector is in $C$ or
$C^\bot$.  In this paper we consider two main types of unitaries that
use this basis.  We show that in one case such a unitary can be used
to solve random subset sum and in the other case it can be used to
solve the random collision problem.  This may indicate that the
unitaries are difficult to implement.

The first type of unitary we consider maps the subset sum basis to the
standard basis.  An example would be one that maps each vector
$|S_{\vecl,p}^m\rangle$ to the corresponding standard basis vector
$|m,p,\vecl\rangle$, identifying the vector.  This unitary can be used
to solve a subset sum instance ($\vecl$, $p$) by taking
$|0,p,\vecl\rangle$, applying $U^{-1}$ to get $|S_{\vecl,p}^0\rangle$
and measuring, since $|S_{\vecl,p}^0\rangle$ is a uniform
superposition of solutions.  The ability to identify the basis vector
in this way is very strong because it can solve an NP-complete
problem, but we show the connection for a wider range of unitaries.
In particular, we show that any unitary that maps the subset sum basis
to the standard basis in some way can be used to solve the random
subset sum problem in the cryptographic range of constant density
$\rho>1$.  This can be view as generalizing the connection to quantum
sampling in~\cite{Bacon06}.

The proof for this case works by showing that such a unitary can be
used to solve worst case collision for the subset sum function.  That
is, given a subset sum instance $(\vecl, p)$ and a solution vector
$\vecb$, the goal is to compute a second solution $\vecb'$ if one
exists.  Then we use the fact that random subset sum reduces to random
collision for density a constant greater than one~\cite{Impagliazzo96}.

The second type of unitary we consider maps the subset sum basis to
vectors where the first bit is zero if the vector is in $C$, and is
one if the vector is from $C^\perp$.  This type of unitary can be used
to solve the DCSP by computing the unitary on the input vector and
measuring the first bit.  It is a relaxation of the first type
of unitary because it could be followed by another unitary mapping to
the standard basis.  We show that this type of unitary can be used to
solve the random collision problem for subset sum with density
$\rho = 1 + c\log \log N/\log N$.  This collision problem for this
density appears to be less well understood than for constant density.

The proof for this case uses the unitary that can solve the DCSP to
solve the random collision problem for subset sum.  The problem in
this case has an arbitrary solution vector $\vecb$ fixed, and then a
vector $\vecl$ is chosen at random.  The goal is again to find a
second solution $\vecb' \neq \vecb$ such that
$\vecb'\cdot \vecl = \vecb \cdot \vecl \bmod N$ on input $\vecb$ and $\vecl$.

In addition to these two main types of unitaries we show that a small
generalization of the form of the subset sum bases has similar results.  

The hardness of random subset sum depends on the density and the same
is true for the collision problem.  But for collision the definition
is important also.  There are four definitions of finding collisions
of hash functions~\cite{Rogaway2004}.  Our definition of random subset
sum collision is based on the universal one-way hash-function family.
That is, for any point in the domain, given the hash function
uniformly at random from the family, the goal is to find another point
in the domain having the same hash value.  Impagliazzo and Naor have
shown that random subset sum collision is at least as hard as the
random subset sum problem when the density is a constant greater than
$1$~\cite{Impagliazzo96}. However, the density of the random subset
sum collision problem we consider has density
$\rho \leq 1+c\log N\log N/\log N$.  This density is between the density
used for subset sum in~\cite{Regev04} and the cryptographic one.  The
hardness of densities for collision in this range is not known, but it
can be contrasted with random subset sum, where the problem gets
harder as the density approaches one~\cite{Impagliazzo96}.

There are several open questions. Can the second type of unitary above
also be used to solve
random subset sum? Consider unitaries which decide membership of $C$
with small error, e.g., $1/poly$.  Can these unitaries be implemented
efficiently or solve some hard problems?  Is it possible to implement
a unitary efficiently distinguishing $C$ from $C^\bot$, with the
subset sum basis, or some other basis?  If a space has a basis that
seems hard to be implemented for some reason, does that mean that no
basis for that space is efficient?  Is it possible that a larger space
$C'$ containing $C$ exists where it is easier to test $C'$ vs.\
$C'^\bot$?  Deciding membership in a subspace or its complement is a
generalization of classical languages to quantum languages.  Are there
other examples?

\section{Background}

In this section, we give the background of the dihedral coset problem and the random subset sum problem. 

The dihedral coset problem comes from the dihedral hidden subgroup problem which is 
\begin{definition}[Dihedral Hidden Subgroup Problem]
 Given the dihedral group $D_{2N}$ and a function $f$ that maps $D_{2N}$ to some finite set such that $f$ hides a subgroup $H$ ($f$ takes same value within each coset of $H$ and takes distinct value on different cosets), the problem is to find a set of generators for $H$.
\end{definition}
Ettinger and Hoyer showed that the problem reduces to the case
when the subgroup is order two~\cite{Ettinger00}. Hence, we can assume
$H$ is an order two subgroup, which can be represented as $\{1,d\}$
for $d\in \mathbb{Z}_N$.  All known approaches for solving this
problem start by evaluating the function in superposition to get
$\sum_{g\in D_{2N}}|g,f(g)\rangle$, and then measuring the function
value.  This results in an order two coset state
$\frac{|0,x\rangle+|1,x+d\rangle}{\sqrt{2}}$, where $x\in \Z_N$ is a
random coset representative.  Then the problem becomes to find $d$ when
given many random order two coset states. This problem is defined as
follows:
\begin{definition}[Dihedral Coset Problem (DCP)] 
	Given a random $k$-register order two coset state 
	\[
		|c_{x_1,x_2,\dots,x_k}^{(d)}\rangle=\frac{1}{\sqrt{2}}(|0,x_1\rangle+|1,x_1+d\rangle)\otimes\cdots\otimes\frac{1}{\sqrt{2}}(|0,x_k\rangle+|1,x_k+d\rangle).
	\]
	The problem is to find $d$.
\end{definition}

The hardness of the DCP has been studied by reducing to the random
subset sum problem~\cite{Regev04} which is defined as follows:  

\begin{definition}[Random Subset Sum Problem] 
  Given a vector of positive integers $\vec{l}=[l_1,l_2,\dots,l_k]^T$
  uniformly distributed in $\mathbb{Z}_N^k$ and
  $s=\vec{b} \cdot \vec{l} \pmod N $ where $\vec{b}\in \mathbb{Z}_2^k$
  is chosen uniformly at random, find a vector
  $\vec{b}'\in \mathbb{Z}_2^k$ such that
  $\vec{l}\cdot\vec{b}'=s \pmod N$. The density is defined
  as $\rho=\frac{k}{\log(N)}$. 
\end{definition}

Although the worst-case subset sum problem is NP-hard, the random
subset sum problem can be solved in polynomial time when the density
is in a certain range. There is no known polynomial-time algorithm for
solving the case when $\rho$ is $\Omega(1/k)$ and
$O(\frac{k}{\log^2(k)})$. Regev~\cite{Regev04} showed that a solution
to the random subset sum problem with $\rho>1$ implies an efficient
quantum algorithm for solving the DCP.  Moreover, we note that one can
reduce the random subset sum problem with $\rho=O(1/\log k)$ to a
lattice problem~\cite{Lag85,Coster92}, and then to the
DCP~\cite{Regev04}. Since these two ranges are generally believed not
equivalent, it is still not clear if the DCP is equivalent to 
random subset sum with $\rho$ in a hard range.

In the rest of this section, we define one more problem which will be used in the section~\ref{sec:hardnessresults}. 
\begin{definition}[Random Subset Sum Collision Problem]\label{def:randcollision}
  Let $\vecb \in \mathbb{Z}_2^k$ be an arbitrary fixed vector.  Given $\vecb$, and a
  vector $\vecl\in\mathbb{Z}_N^k$ chosen uniformly at random, the problem
  is to find a solution
  $\vec{b}'\in \mathbb{Z}_2^k$ such that
  $\vec{b}\cdot \vecl\equiv \vec{b}'\cdot \vecl \pmod N$ and
  $\vec{b}'\neq \vec{b}$.
\end{definition}

The worst-case version of this problem is to find $\vec{b}'$ for
arbitrary $\vec{b}$ and $\vec{l}$ which are given. For simplicity, we
will call this problem the random collision problem and the worst-case
version as the worst-case collision problem in the rest of the paper.

Impagliazzo and Naor showed a relationship between random collision
problem and the random subset sum problem.  The input in their
notation has $n$ numbers modulo $2^{\ell(n)}$ plus the target value.

\begin{theorem}[Theorem 3.1 in~\cite{Impagliazzo96}]
  \label{thm:rssp_rsscp_eqv}
  Let $\ell(n)=(1-c)n$ for $c>0$. If the subset sum function for
  length $\ell(n)$ is one-way, then it is also a family of universal
  one-way hash functions.  
\end{theorem}

The subset sum function for length $\ell(n)$ can be represented by $n$
integers each of which is $\ell(n)$-bits long. The input is an $n$-bit
binary string $\vecb$ which indicates a subset of the $n$ integers and
the function outputs an integer $s$ which is the sum of the subsets of
integers indexed by $\vecb$. A family of universal one-way hash
functions is the set of functions $\mathcal{F}=\{f\}$ which satisfies
the property that if for all $x$, when $f$ is chosen randomly from
$\mathcal{F}$, then finding a collision (i.e., $y\neq x$ and
$f(x)=f(y)$) is hard. Note that the random subset sum problem can be
viewed as inverting a random subset sum function and the random
collision problem is as finding a collision for a random subset sum
function.

In the proof of Theorem~\ref{thm:rssp_rsscp_eqv}~\cite{Impagliazzo96},
Impagliazzo and Naor showed that finding a collision for a random
subset sum function is at least as hard as inverting a random subset
sum function. Therefore, we can give the following corollary:  

\begin{corollary}\label{cor:rssp_rsscp_eqv}
  The random subset sum problem with $N$ a power of $2$ and $\rho$ a constant $>1$
  reduces to the  random collision problem with the same $N$ and $\rho$.
\end{corollary}

This corollary will be used in the section~\ref{sec:hardnessresults}.

\section{The Dihedral Coset Space Problem}

In this section we set up our approach.  We first define the dihedral
coset space problem and show how to use it to solve the dihedral coset
problem.  Then we define the coset space which we wish to understand.

\begin{definition}[Dihedral Coset Space Problem (DCSP)]
  Given a state $|\tau\rangle$ which is promised to be a random
  order-two coset state $|c^{(d)}_{x_1,x_2,\dots,x_k}\rangle$ or a random standard basis state
  $|\vecb,x\rangle$ where $\vecb\in \mathbb{Z}_2^k $ and
  $x\in \mathbb{Z}_N^k$, the problem is to decide if $|\tau\rangle$ is
  a $k$-register order two coset state or not.
\end{definition}

A solution to the DCSP implies a polynomial-time algorithm for solving
the DCP with $N$ a power of $2$ as shown in~\cite{FZ08}.   We include
a proof of our special case here.

\begin{claim}
  The dihedral coset problem (DCP) with $N$ a power of $2$ reduces to the dihedral coset space problem (DCSP).
\end{claim}
\begin{proof}
  Suppose we are given the input of the DCP with subgroup $d$, we
  first show how to get the least significant bit of $d$. 
  
  Since $N$ is a power of $2$, the least significant bit of $x$ and $x+d \pmod N$ are equal for $x\in \mathbb{Z}_N$ if and only if $d$ is even.
  Therefore by measuring the least significant bit of the state $\frac{|0,x\rangle+|1,x+d\rangle}{\sqrt{2}}$, we get the same state if $d$ is even and 
  get either $|0,x\rangle$ or $|1,x+d\rangle$ (which are standard-basis states) otherwise. 
  
  According to the observation above, the least significant bit of $d$ can be computed by the following algorithm. 
  First, we measure the least significant bit of each register. Then all the registers do not 
  change or collapse to a standard-basis state. Finally, apply the algorithm for the DCSP. 
  If the result is an order-two coset state, the least significant bit is $0$; otherwise, the least significant bit is $1$.

 To get bit $(i+1)$, one subtracts $d$ by the least significant
 $i$ bits computed and measure the $I+1$-th least significant bit of the state. 
 Repeat the process above until all bits of $d$ are known.
\end{proof}
It is worth noting that this fact also implies that the lattice problem can be reduced to the DCSP due to
the known reduction from the lattice problem to the DCP with $N$ a power of $2$~\cite{Regev04}.

The main objects we want to understand are the coset space and its complement.
\begin{definition}
  The coset space $C=\Span(\{|c_{x_1,\dots,x_k}^{(d)}\rangle:\:
  d,x_1,\dots,x_k\in \mathbb{Z}_N\})$ and the orthogonal complement of
  $C$ is $C^{\bot}$.
\end{definition}
Note that a test for a vector being in $C$ or $C^\bot$ is sufficient
to solve the DCSP if $k$ is big enough.  This follows from counting the number of
$k$-register order two coset states. There are at most $N$ subgroups,
and at most $N^k$ coset representatives, so the number of $k$-register
order two coset states is at most $N(N)^k$.  The dimension of the
whole space is $(2N)^k$. Hence, the subspace spanned by $k$-register
order-two coset states is at most $1/2$ of the whole space when
$k\geq\log 2N$.   

\begin{claim}\label{claim:solvedcsp}
  Let $k=\log 2N+k'$.  Let $\Pi_{C}$ be a projector onto $C$ and
  $\Pi_{C^{\bot}}$ be a projector onto $C^{\bot}$. If the input is an
  order two coset state, the measurement $\{\Pi_C, \Pi_{C^{\bot}}\}$
  outputs $C$ always. Otherwise, if the input is a random standard
  basis state, then this measurement outputs $C^{\bot}$ with
  probability at least $1-1/2^{k'+1}$.
\end{claim}

\section{The Subset Sum Basis}\label{sec:subsetsumbasis}

In this section, we start by finding an orthonormal basis for $C$ and
one for $C^{\bot}$. Note that if we can give a unitary which distinguishes 
which of the two subspaces we are in ($C$ or $C^{\bot}$)
efficiently, we can solve the DCSP efficiently as in
Claim~\ref{claim:solvedcsp}.  

In order to make the basis easier to understand, we permute the
subsystems so that the first bit of all registers are on the left, and
the integers mod $N$ are on the right. That is, write the original
basis state $|b_1,x_1,b_2,x_2,\dots,b_k,x_k\rangle$ as
$|b_1,b_2,\cdots,b_k,x_1,x_2,\cdots, x_k\rangle$.  In this notation
the coset state is written
as
\[
|c^{(d)}_{x_1,x_2,\dots,x_k}\rangle = \frac{1}{2^{k/2}}
\sum_{b_1,\dots,b_k=0}^1 |b_1,\dots,b_k,x_1+b_1 d,\dots, x_k+b_k
d\rangle = \frac{1}{2^{k/2}} \sum_{\vecb \in
  \{0,1\}^k}|\vecb,\vec{x}+\vecb d\rangle.
\]

The subset sum basis is defined as follows: 
\begin{definition} [The Subset Sum Basis]\label{def:basis}
  Let $\vecl=(l_1,l_2,\cdots,l_k)^T\in \mathbb{Z}_N^k ,$ and
  $p \in \Z_N$.  Let
  $T_{\vecl,p} =\{\vec{b}: \vec{b}\cdot \vecl = p, \vec{b} \in
  \mathbb{Z}_2^k\}$ contain subset sum solutions for input $\vecl$,
  $p$.
  If $|T_{\vecl,p}|=0$ then define $|S_{\vecl,p}^m\rangle = 0$.  If
  $|T_{\vecl,p}|\geq 1$, then let $m\in \{0,\dots, |T_{\vecl,p}|-1\}$
  and pick an ordering $\{\vec{b}^{(j)}_{\vecl,p}\}$ of the solutions
  in $T_{\vecl,p}$.  Define the vector
\begin{equation}\label{eq:reflection}
  |S_{\vecl,p}^m\rangle = \tfrac{1}{\sqrt{|T_{\vecl,p}|}} \sum\limits_{j=0}^{|T_{\vecl,p}|-1}
  \omega_{|T_{\vecl,p}|}^{mj}| \vec{b}^{(j)}_{\vecl,p}\rangle.
\end{equation}
For $N, k \in \Z$, define two sets
\begin{equation} \label{eq:ssbasisperp} 
  \basisname^\bot =
  \basisname_{k,N}^{\bot}=
  \{|S_{\vecl,p}^{m}\rangle|\chi_{\vecl}\rangle: \vecl \in
  \mathbb{Z}_N^k,\, p\in \mathbb{Z}_N,\, m\in
  \{1,\dots,|T_{\vecl,p}|-1\},\, |T_{\vecl,p}|\geq 2\}
\end{equation}
and 
\begin{equation} \label{eq:ssbasis0}
	\basisname^0=\basisname_{k,N}^0=\{ |S_{l,p}^{m}\rangle|\chi_{\vecl}\rangle:
	\vecl\in \mathbb{Z}_N^k,\, p\in \mathbb{Z}_N,\,
        m=0,\, |T_{\vecl,p}|\geq 1\}.
\end{equation}
The set
$\basisname=\mathcal{B}^0 \bigcup
\mathcal{B}^{\bot}$
is called the {\em subset sum basis} of $\mathbb{C}^{(2N)^k}$.
\end{definition}

In this definition, $|\chi_j\rangle$ is the Fourier basis state
$|\chi_j\rangle = \tfrac{1}{\sqrt{N}}\sum_i \omega_N^{ij} |i\rangle$,
and
$|\chi_{\vecl}\rangle = |\chi_{l_1}\rangle\cdots |\chi_{l_k}\rangle$.
Note that $\mathcal{B}^0 \cup \mathcal{B}^\perp$ is an orthonormal
basis for the whole space and the two sets are disjoint. The vector
$|S_{\vecl,p}^{m}\rangle$ is a superposition of solution vectors
$\vecb$ to the equation $\vecl \cdot \vecb = p$. If no such $\vecb$
exists then there is no corresponding $|S_{\vecl,p}^m\rangle$. If at
least one solution $\vecb$ exists then $|S^0_{\vecl,p}\rangle$ is in
$\basisname^0$. If at least two solutions $\vecb$ exist then vectors
appear in $\basisname^\perp$. Varying $m$ gives orthogonal
superpositions of the solutions. Ranging over all $\vecl \in \Z_N^k$ and
$p\in \Z$ covers all possible bit vectors. Furthermore, these vectors
are tensored with every possible Fourier basis state over $\Z_N$.

Next we show that $\basisname^{\bot}$ forms an orthonormal basis for $C^{\bot}$.

\begin{claim}\label{claim:ssb_orthogonal}
  The vectors in the set $\basisname^{\bot}$ form an orthonormal
  basis of a space that is orthogonal to the $k$-register order two
  coset space. 
\end{claim}

\begin{proof}
  As noted, the vectors form an orthonormal basis of the whole space.
  We will show that an arbitrary state
  in $\basisname^{\bot }$ is orthogonal to all $k$-register order
  two coset states. Fix $\vecl$ and $p$, and let
\[
|\psi\rangle=
|S_{\vecl,p}^m\rangle |\chi_{\vecl}\rangle
= 
\tfrac{1}{\sqrt{|T|}} \sum\limits_{j=0}^{|T|-1} \omega_{|T|}^{mj} |\vec{b}^{(j)}\rangle|\chi_{\vecl}\rangle
\]
be a
  state in $\basisname^{\bot}$  where $T=T_{\vecl,p}$ and
  $\vec{b}^{(j)}=\vec{b}^{(j)}_{\vecl,p}$ to simplify notation. Then for
  an arbitrary order-two coset state $|c_{x_1,x_2,\cdots,x_k}^d\rangle$, the inner product
  $\langle c^{(d)}_{x_1,x_2,\cdots,x_k}|\psi\rangle$ is
\begin{eqnarray} 
  \frac{1}{\sqrt{2^k |T|}}  \sum_{\vecb\in \{0,1\}^k}
  \langle\vecb|\langle \vec{x}+\vec{b}d| \sum\limits_{j=0}^{|T|-1}
  \omega_{|T|}^{mj}|\vec{b}^{(j)}\rangle|\chi_{\vecl}\rangle \nonumber 
  &=& \frac{1}{\sqrt{2^kN^k |T|}}  \sum\limits_{j=0}^{|T|-1}
      \omega_{|T|}^{mj}
      \omega_{N}^{\vecl\cdot(\vec{x}+d\vec{b}^{(j)})} \nonumber\\ 
  &=& \frac{\omega_N^{\vecl\cdot\vec{x}}}{\sqrt{2^kN^k |T|}}	\sum\limits_{j=0}^{|T|-1} \omega_{|T|}^{mj} \omega_{N}^{dp} \label{eq:p}\\
  &=&\frac{\omega_N^{\vecl\cdot\vec{x}+dp}}{\sqrt{2^kN^k |T|}}	\sum\limits_{j=0}^{|T|-1} \omega_{|T|}^{mj} =0 \label{eq:zero}.
\end{eqnarray} 

Eq.~\ref{eq:p} is true because $\vec{b}^{(j)}\cdot \vecl=p$ iff
$\vecb^{(j)}\in T$ by the definition of $T$. Then since $m\geq 1$ and
$|T|\geq 2$
by the definition of $\basisname_{k,N}^{\bot}$, Eq.~\ref{eq:zero} is
true.
\end{proof}

According to Claim~\ref{claim:ssb_orthogonal}, $\Span(\basisname^{\bot})\subseteq C^{\bot}$. Next we show that $\basisname^0$ exactly spans the subspace $C$ (and thus $\Span(\basisname^{\bot})=C^{\bot}$). 

\begin{lemma}\label{lem:cosetspace}
  The set $\basisname^0$ is an orthonormal basis for the
  subspace spanned by the order-two coset states.
\end{lemma}

\begin{proof}
  Because $C$ is orthogonal to $\Span(\basisname^\perp)$ by
  Claim~\ref{claim:ssb_orthogonal}, $C\subseteq \Span(\basisname^0)$.
  We want to show equality.  Suppose for contradiction that $C\subset
  \Span(\basisname^0)$.  Then there is a vector $|\alpha\rangle \in C^\perp$ that
  is orthogonal to $\Span(\basisname^\perp)$, so $|\alpha\rangle \in C^\perp \cap
  \Span(\basisname^0)$. We show that there is no non-zero linear
  combination of states in $\basisname^0$ whose inner product
  with all order-two coset states is zero. 
  
  Suppose the state
  \[
  |\alpha\rangle = \sum_{\vecl\in \mathbb{Z}_N^k,p\in \mathbb{Z}_N}
  \alpha_{\vecl,p}
  |S_{\vecl,p}^0\rangle|\chi_{l_1},\dots,\chi_{l_k}\rangle
  \]
  is orthogonal to all order-two coset states, i.e.,
  $\langle c^{(d)}_{x_1,x_2,\dots,x_k}|\alpha\rangle=0$ for
  $x_1,\dots,x_k, d\in \mathbb{Z}_N$, for some nonzero vector $|\alpha\rangle \in \Span(B^0)$.  This inner product is

  \begin{eqnarray*}
    &&
       \frac{1}{\sqrt{2^k}}\sum_{\vecb \in \{0,1\}^k} \langle
       \vecb,\vec{x}+\vecb d|  \sum_{\vecl\in \mathbb{Z}_N^k,p\in
       \mathbb{Z}_N} \alpha_{\vecl,p} 
       |S_{\vecl,p}^0\rangle|\chi_{l_1},\dots,\chi_{l_k}\rangle\\
    &=&
        \frac{1}{\sqrt{2^k}}\sum_{\vecb \in \{0,1\}^k} \langle 
        \vecb,\vec{x}+\vecb d|  \sum_{\vecl\in \mathbb{Z}_N^k,p\in 
        \mathbb{Z}_N} \alpha_{\vecl,p} 
         \frac{1}{\sqrt{|T_{\vecl,p}|}} \sum_{j=0}^{|T_{\vecl,p}|-1}
       |\vecb_{\vecl,p}^{(j)}  \rangle|\chi_{l_1},\dots,\chi_{l_k}\rangle  \\
    &=&
        \frac{1}{\sqrt{2^kN^k}}
        \sum_{\vecl\in \mathbb{Z}_N^k,p\in  \mathbb{Z}_N}
        \alpha_{\vecl,p} \frac{1}{\sqrt{|T_{\vecl,p}|}}
        \sum_{\vecb \in \{0,1\}^k} \sum_{j=0}^{|T_{\vecl,p}|-1} \langle \vecb|
        b_{\vecl,p}^{(j)}\rangle  \omega_N^{\vecl\cdot  (\vec{x}+\vecb d)}\\
    &=&
        \frac{1}{\sqrt{2^kN^k}}
        \sum_{\vecl\in \mathbb{Z}_N^k,p\in  \mathbb{Z}_N}
        \alpha_{\vecl,p} \frac{\omega_N^{\vecl\cdot\vec{x}+pd}}{\sqrt{|T_{\vecl,p}|}}
        \sum_{\vecb : \vecb \cdot \vecl = p} \sum_{j=0}^{|T_{\vecl,p}|-1} \langle\vecb|
        b_{\vecl,p}^{(j)}\rangle \\
    &=&
        \frac{1}{\sqrt{2^kN^k}}
        \sum_{\vecl\in \mathbb{Z}_N^k,p\in  \mathbb{Z}_N}
        \alpha_{\vecl,p} \omega_N^{\vecl\cdot\vec{x}+pd}\sqrt{|T_{\vecl,p}|}.
  \end{eqnarray*}

 Then we have the following equations:
\[
  \sum_{\vecl\in \mathbb{Z}_N^k,p\in \mathbb{Z}_N}
  \sqrt{\frac{|T_{\vecl,p}|}{2^kN^k}} \alpha_{\vecl,p} \cdot
  \omega_N^{x_1\cdot l_1+\cdots +x_k\cdot l_k + d\cdot p} = 0,\,
  \forall x_1,\dots,x_k,d \in \Z_N.
\]

Define $\vec{v}$ as an $N^{k+1}\times 1$ vector such that
$\left(\vec{v}\right)_{\vecl,p} =\sqrt{\frac{|T_{\vecl,p}|}{2^kN^k}} \alpha_{\vecl,p}$.
The sums above can be represented as follows: 
\begin{eqnarray}\label{eq:Av0}
	A^{\otimes (k+1)}\cdot \vec{v} = \vec{0}, 
\end{eqnarray}
where $A$ is an $N\times N$ Fourier matrices with the $(i,j)$-th entry as $A_{i,j} = \omega_N^{i j}$ 
and $\vec{0}$ is an $N^{k+1}\times 1$ vector with all entries as $0$. Note that the column of $A^{\otimes (k+1)}$ 
is indexed by $\vecl$ and $p$ and the the row is indexed by $\vec{x}$ and $d$. 

The determinant of $A^{\otimes(k+1)}$ 
is not zero, so the only vector $\vec{v}$ satisfying Equation~\ref{eq:Av0} 
is $\vec{v}=\vec{0}$.  When $|T_{\vecl,p}|\geq 1$ this forces
$\alpha_{\vecl,p} =0$ for every coefficient used in $|\alpha\rangle$.
When $|T_{\vecl,p}|=0$, $\alpha_{\vecl,p}$ is not used in the sum
because $|S_{\vecl,p}^m\rangle =0$. Therefore, these facts contradict the hypothesis that
there exists a nonzero vector $|\alpha\rangle \in \Span(\basisname^0)$ which is orthogonal to
all order-two coset states. 
\end{proof}

Now, it is easy to see that a unitary which can efficiently
distinguish $\Span(\basisname^0)$ from 
$\Span(\basisname^{\bot})$ also distinguishes $C$ from $C^{\bot}$ by
Claim~\ref{claim:solvedcsp} and Lemma~\ref{lem:cosetspace}.  
The next question we address is whether any unitaries that use this
basis can be implemented efficiently or not.

\section{The hardness results}\label{sec:hardnessresults}

In general we would like to understand unitaries that can be used to
decide if a state is in the coset space $C$ or in $C^\bot$. In this
section we look at two types of unitaries using the subset sum basis,
plus an extension of each one:
\begin{enumerate}
\item A unitary $\Us$ that maps every basis vector
  $|S_{\vecl,p}^m\rangle|\chi_\vecl\rangle$ to a standard basis state.
  Note that if these standard basis states specify $p$ and $\vecl$, then this can be
  used to solve the worst case subset sum, but we are allowing a more
  general type of unitary here.
\item A unitary $\Uc$ that maps every basis vector $|S_{\vecl,p}^m\rangle|\chi_\vecl\rangle$ to
  $|m= 0?\rangle|\phi_{\vecl,p}^m\rangle$, indicating whether or
  not the state is in the coset space.

\item A unitary $U=\tUs$ that satisfies condition (1) or $U=\tUc$ that
  satisfies (2), but $U$ uses a
    slightly more general basis, where any basis can be chosen for each
    $(\vecl,p)$ subspace $\Span\{|S_{\vecl,p}^m,\chi_{\vecl}\rangle:
    m \geq 1\}$.
\end{enumerate}

For the last type we use any basis satisfying the following
definition.  

\begin{definition}
  \label{def:newbasis}
  Let
  $\tilde{\basisname}^0=\basisname^0
  =\{|S_{\vecl,p}^0\rangle|\chi_{\vecl}\rangle : \vecl \in \Z_N^k,
  p\in \Z_N, m=0, |T_{\vecl,p}|\geq1 \}$
  be as in Definition~\ref{def:basis}, and let
  $\tilde{\basisname}^\perp=\{|\tilde{S}_{\vecl,p}^m\rangle
  |\chi_\vecl\rangle : \vecl\in \Z_N^k, p \in \Z_N, m\in
  \{1,\dots,|T_{\vecl,p}|-1\}, |T_{\vecl,p}|\geq 2\}$
  be an orthogonal basis such that
  $\Span(\{|\tilde{S}_{\vecl,p}^m \rangle : m\in
  \{1,\dots,|T_{\vecl,p}|-1\} \}) =\Span(\{|S_{\vecl,p}^m \rangle :
  m\in \{1,\dots,|T_{\vecl,p}|-1\} \})$
  for all $\vecl, p$.
\end{definition}

We show that unitaries of type 1 above can be used to solve random
subset sum for the cryptographic density $\rho$ a constant greater
than $1$, indicating that such a unitary may be hard to implement.
This strengthens the result in~\cite{Bacon06} which is a special case
where the unitary must perform quantum sampling, i.e., map an input
$|\vecl,p\rangle$ to a superposition of solutions
$|\vecl,S_{\vecl,p}^0\rangle$.  Such a unitary implementing quantum
sampling can used to solve worst-case subset sum by taking an input
$|0,p,\vecl\rangle$, applying $U$ inverse to get
$|S_{\vecl,p}^0\rangle|\vecl\rangle$ and measuring, since this is a
uniform superposition of solutions.

An algorithm that uses the subset sum basis to solve the DCSP needs to
decide if $m=0$ (for $C$), or $m>0$ (for $C^\bot$).  The second type
of unitary above allows an arbitrary unitary that writes the answer in
the first bit.  We show that such a unitary can solve random collision
for density $\rho=1+c\log \log N/\log N$.  This may indicate that no
such unitary can be efficiently implemented, although we are less
clear on the difficulty of the random collision problem.

The third type of unitary allows an arbitrary basis within each
subspace of solutions, but does not mix solutions of different inputs
$\vecl, p$.  Note that $|S_{\vecl,p}^0\rangle$ cannot change in this
case, since it is one dimension in $\basisname^0$.  Let
$\tilde{\basisname} = \tilde{\basisname}^0\cup
\tilde{\basisname}^\perp$ be the basis used by the unitary.

The proofs work by using $\Us$ to solve the worst-case collision
problem, or $\Uc$, $\tUs$, or $\tUc$ to solve the random collision
problem.

\subsection{Unitary mapping to a standard basis}\label{subsec:ssb_standard}

First we give an algorithm that finds a solution to the worst-case
collision problem when given a unitary $\Us$ that maps the subset
sum basis $\basisname$ to the standard basis in an arbitrary way.  Given
$\vec{b}$ and $\vec{l}$ where $\vec{b}\in \mathbb{Z}_2^k$ and
$\vec{l}\in \mathbb{Z}_N^k$, the task in the worst-case collision
problem is to find $\vecb'\neq \vecb$ such that
$\vecb'\cdot\vec{l}=\vec{b}\cdot \vec{l}$.

\begin{algorithm}\label{alg:register}
On input $\vecl\in \mathbb{Z}_N^{k}$ and $\vecb\in \mathbb{Z}_2^k$:
\begin{enumerate}
	\item Prepare the quantum state $|\vec{b},\, \vecl\rangle$.
	\item Apply $QFT_{N}^k$ on $\vecl$, then the state becomes $|\vec{b}\rangle|\chi_{\vecl}\rangle$.
	\item Apply $\Us$ to $|\vec{b}\rangle|\chi_{\vecl}\rangle$.
	\item Measure $\Us(|\vec{b}\rangle|\chi_{\vecl}\rangle)$ in the standard basis.
	\item Apply $\Us^{\dagger} $.
	\item Measure value $\vecb'$ in the first register.
\end{enumerate}
\end{algorithm}
Here $QFT_{N}^k$ is the quantum Fourier transform over $\mathbb{Z}_N^{k}$.

\begin{theorem}\label{lem:ssb_to_standard}
  If there exists an efficient unitary operator $\Us$, where $\Us$ is a
  bijection between the subset sum basis and the standard basis, then
  the worst-case collision problem can be solved efficiently by a
  quantum algorithm.
  Therefore random subset sum with density a
  constant greater than 1 can also be solved.
\end{theorem}

\begin{proof}
  Given $\vecl$ and $\vecb$ as input, let $p=\vecl\cdot\vecb$ and
  $T=T_{\vecl,p}$. For $\vecb =\vecb^{(j_0)} \in T$, after computing
  the Fourier transform of the second register, the
  resulting state 
  $|\vecb^{(j_0)},\chi_\vecl\rangle$ can be wrtten in the subset
  sum basis as
  \[
  |\vecb^{(j_0)},\chi_\vecl\rangle = \frac{1}{\sqrt{|T|}}
  \sum_{m=0}^{|T|-1}\omega^{-j_0m}|S_{\vecl,p}^{m}\rangle|\chi_\vecl\rangle.
  \]
  Applying $\Us$ to this state gives the state
  \begin{equation}
    \label{eq:collision_a_U}
    \frac{1}{\sqrt{|T|}}\sum_{m=0}^{|T|-1}\omega^{-j_0m}|D_{\vecl,p}^m\rangle,
  \end{equation}
  where $|D_{\vecl,p}^m\rangle:=\Us(|S_{\vecl,p}^m\rangle|\chi_\vecl\rangle)$ is
  a standard basis vector by assumption on $\Us$. Measuring the
  state in Equation~(\ref{eq:collision_a_U}) in the standard basis
  gives $|D_{\vecl,p}^m\rangle$ for some 
  $m\in [0:|T|-1]$. Applying $\Us^{\dagger}$ to
 $|D_{\vecl,p}^m\rangle$ gives
 $|S_{\vecl,p}^m\rangle|\chi_\vecl\rangle$, where the first register is
 $|S_{\vecl,p}^m\rangle= \frac{1}{\sqrt{|T|}}
   \sum_{j=0}^{|T|-1}\omega^{jm}|\vecb^{(j)}\rangle$
   in the standard basis.  Measuring this gives a vector $\vecb'\neq \vecb$ with
   probability $\frac{|T|-1}{|T|}$. 

   Theorem~\ref{thm:rssp_rsscp_eqv} reduces random subset sum 
   to solving the random collision problem for constant density
   greater than one, so random subset sum
   also reduces to the worst case collision problem.
\end{proof}

The proof that Algorithm~\ref{alg:register} works used a special
property of the subset sum basis, 
which is that every basis vector $|S_{\vecl,p}^m\rangle$ spreads the
solutions with equal magnitude.  When this is not the case then
the algorithm does not work for the worst case collision problem.  However, we will later
show that it solves the random collision problem, as long as the
number of solutions is not too large.

First we describe an example basis where the algorithm fails.  The idea is
that the unitary can map a solution 
vector $|\vecb, \chi_\ell\rangle$ to a vector that is very close to itself.
In that case the algorithm will measure the same value $\vecb$ that it started
with and not solve the collision problem, which can be seen as
follows.  Let $\vecb = \vecb^{(0)}$, let 
\[
    |\hat{S}^{1}\rangle|\chi_\vecl\rangle = \frac{1}{\sqrt{|T|-1}}\sum_{m=1}^{|T|-1}|S_{\vecl,p}^{m}\rangle|\chi_\vecl\rangle,
\]
and pick arbitrary orthonormal vectors
$|\hat{S}^{2}\rangle,\dots,|\hat{S}^{|T|-1}\rangle$ to form a basis for the subspace
$\Span(\{|S_{\vecl,p}^{m}\rangle|\chi_\vecl\rangle: m\in
[1:|T|-1]\})$.
Note that
$|\langle\vecb,\chi_\vecl|\hat{S}^{1}, \chi_\vecl\rangle|^2 =
\frac{|T|-1}{|T|}$,
which implies that one gets $\Us(|\hat{S}^{1}\rangle |\chi_\vecl\rangle)$ 
with probability $\frac{|T|-1}{|T|}$ after applying $\Us$ and
measuring in the standard basis.  In that case, applying
$\Us^{\dagger}$ results in the input vector $\vecb$.  Therefore, given
a unitary mapping this new basis
$\{|S^{0},\chi_\vecl\rangle, |\hat{S}^{1},
\chi_\vecl\rangle,\dots,|\hat{S}^{|T|-1}, \chi_\vecl\rangle\}$
to standard basis, the algorithm returns an answer $\vecb'\neq \vecb$
happens with probability $1/|T|$.  The number of solutions $|T|$ can
be very large for larger densities.

Next we show that if we limit the size of $T$, then random collision
can be solved.

\begin{corollary}\label{lem:arbitrary_to_standard}
  Suppose Algorithm~\ref{alg:register} is run with $\tUs$.  If $\tUs$
  is an efficient unitary operator which maps every state in
  $\tilde{\basisname}$ to an arbitrary state in the standard basis,
  then on input $\vecl, \vecb$, the algorithm solves the collision problem with
  probability at least
  $\frac{1}{|T_{\vecl,p}|}(1-\frac{1}{|T_{\vecl,p}|})$, where
  $p=\vecl \cdot \vecb$.  In particular, when
  $k\leq \log N + c \log \log N$, the random collision problem can be
  solved in quantum polynomial time.
\end{corollary}

\begin{proof}

  Similar to the proof
  for Theorem~\ref{lem:ssb_to_standard}, first represent
  $|\vecb,\chi_\vecl\rangle$ as a linear combination of states in
  $\tilde{\basisname}$ as follows:
\[
    |\vecb,\chi_\vecl\rangle = \frac{1}{\sqrt{|T|}}|S^0\rangle|\chi_\vecl\rangle + \sqrt{\frac{|T|-1}{|T|}}(\sum_{m=1}^{|T|-1}c_m|\tilde{S}^m\rangle)|\chi_\vecl\rangle,  
\]
where $T=T_{\vecl,p}$ and $\tilde{S}^m=\tilde{S}_{\vecl,p}^m$. 

After applying the unitary $\tUs$, the state is
\begin{equation}\label{eq:tildeU_standard}
    \tUs|\vecb,\chi_\vecl\rangle = \frac{1}{\sqrt{|T|}}|D^0\rangle + \sqrt{\frac{|T|-1}{|T|}}(\sum_{m=1}^{|T|-1}c_m|D^m\rangle),
\end{equation}
where $D^m$ for $m\in [0:|T|-1]$ are arbitrary distinct states in the standard basis. 
By measuring the state in the Equation~(\ref{eq:tildeU_standard}) in
the standard basis, $|D^0\rangle$ is measured with probability
$1/|T|$. Then applying $\tUs^{\dagger}$ and measuring the output state
in the standard basis gives $\vecb'\neq \vecb$ with probability
$\frac{|T|-1}{|T|}$. Based on the Claim~\ref{claim:T_bound}, $|T| =
poly(k) $ with high probability. Thus, given a random input $\vecb$
and $\vecl$, the probability to get $\vecb'\neq \vecb$ using $\tUs$ in
Algorithm~\ref{alg:register} is at least $\frac{|T|-1}{|T|^2}=1/poly(k)$.  
\end{proof}

\subsection{Deciding membership in C}\label{subsec:D_C}

The unitary $\Us$ illustrated how our algorithm works and used the
subset sum basis, but $\Us$ may not be useful for distinguishing $C$
from $C^\perp$ in general. Next we consider a unitary $\Uc$ that can distinguish
$C$ from $C^\perp$. Suppose $\Uc$ works on a larger Hilbert space to
have work space and exactly distinguishes $\basisname^0$ from
$\basisname^\perp$ in the first qubit as follows:

\begin{definition} \label{def:one_qubit_m}
  Let $\Uc$ be a unitary operator such that 
  \[
\Uc( |S_{\vecl,p}^m\rangle|\chi_\vecl \rangle |0\rangle) =
 \left\{ \begin{array}{ll}
 	|0\rangle|\psi_{\vecl,p,0}\rangle &\mbox{ if $m=0$}	\\
 	|1\rangle|\psi_{\vecl,p,m }\rangle &\mbox{ otherwise}
  \end{array} \right.
\]
  where
  $\{| \psi_{\vecl,p,m}\rangle: \vecl\in \mathbb{Z}_N^k,\, p\in
  \mathbb{Z}_N,\, m\in [0:|T_{\vecl,p}|-1]\}$
  are states resulting from applying $\Uc$ and the third register is a
  workspace initialized to $|0\rangle$.
\end{definition}

We modify Algorithm~\ref{alg:register} so that only the first bit is
measured in step four, and $\Uc$ is used instead of $\Us$.

\begin{algorithm}\label{alg:bit}
On input $\vecl\in \mathbb{Z}_N^{k}$ and $\vecb\in \mathbb{Z}_2^k$:
\begin{enumerate}
	\item Prepare the quantum state $|\vec{b},\, \vecl\rangle$.
	\item Apply $QFT_{N}^k$ on $\vecl$, then the state becomes $|\vec{b}\rangle|\chi_{\vecl}\rangle$.
	\item Apply $\Uc$ to $|\vec{b}\rangle|\chi_{\vecl}\rangle$.
	\item Measure the first qubit of
          $\Uc(|\vec{b}\rangle|\chi_{\vecl}\rangle)$ in the standard
          basis.
	\item Apply $\Uc^{\dagger} $.
	\item Measure the first register in the standard basis. 
\end{enumerate}
\end{algorithm}

\begin{theorem}\label{lem:relax_image}
  If $\Uc$ can be implemented efficiently, then
  Algorithm~\ref{alg:bit} solves the collision problem on input
  $\vecl, \vecb$ with probability
  $\frac{2}{|T_{\vecl,p}|}(1-\frac{1}{|T_{\vecl,p}|})$, where
  $p=\vecl\cdot \vecb$.  In particular, when
  $k \leq \log N + c \log \log N$ the random collision problem can be
  solved in quantum polynomial time.
\end{theorem}

\begin{proof}
  Given $\vecl$ and $\vecb$ as input, let $p=\vecl\cdot\vecb$ and
  $T=T_{\vecl,p}$. For $\vecb =\vecb^{(j_0)} \in T$, after computing
  the Fourier transform of the second register, we can write the
  resulting state 
  $|\vecb^{(j_0)},\chi_\vecl\rangle$ in the subset
  sum basis as follows:
 \[
  |\vecb^{(j_0)},\chi_\vecl\rangle = \frac{1}{\sqrt{|T|}}
  \sum_{m=0}^{|T|-1}\omega^{-j_0m}|S_{\vecl,p}^{m}\rangle|\chi_\vecl\rangle.
 \]
  Applying $\Uc$ to this state plus a work register results in
  \begin{equation}
    \label{eq:collision_a_U'}
    \frac{1}{\sqrt{|T|}}(|0\rangle|\psi_{\vecl,p,0}\rangle +
    \sum_{m=1}^{|T|-1}\omega^{-j_0m}|1\rangle |\psi_{\vecl,p,m}\rangle).
  \end{equation}
   Measuring the first qubit of the state in
  Equation~(\ref{eq:collision_a_U'}) 
  gives $|0\rangle|\psi_{\vecl,p,m}\rangle$ with probability $1/|T|$
  and
  $\frac{1}{\sqrt{|T|-1}}
  \sum_{m=1}^{|T|-1}\omega^{-j_0m}|1\rangle|\psi_{\vecl,p,m}\rangle$ 
  with probability $1-1/|T|$.
  
  Applying $\Uc^{\dagger}$ to the result gives
  $|S_{\vecl,p}^0\rangle|\chi_\vecl\rangle$ in the first case and
  \[
  \frac{1}{\sqrt{|T|-1}}\sum_{m=1}^{|T|-1}\omega^{-j_0m}|S_{\vecl,p}^m\rangle
  |\chi_\vecl\rangle
  \]
  in the second case.  
 
  Finally, the state is measured in the standard basis. In the first
  case, when a zero is measured in the first bit, which happens with
  probability $1/|T|$, a vector $\vecb'\neq \vecb$ is measured with
  probability $1-1/|T|$ in the last step.  In the second case when a
  one is measured the amplitude of $|\vecb^{(j_0)},\chi_\vecl\rangle$
  in
  $\frac{1}{\sqrt{|T|-1}}\sum_{m=1}^{|T|-1}
  \omega^{-j_0m}|S_{\vecl,p}^m\rangle|\chi_\vecl\rangle$ is
  \[
  \frac{1}{\sqrt{|T|-1}}\sum_{m=1}^{|T|-1}
  \omega^{-j_0m}
  \langle\vecb^{(j_0)},\chi_\vecl |S_{\vecl,p}^m\rangle|\chi_\vecl\rangle
  =
  \frac{1}{\sqrt{|T|-1}}\sum_{m=1}^{|T|-1}
  \omega^{-j_0m}
  \langle\vecb^{(j_0)} |S_{\vecl,p}^m\rangle 
  \]
  \[
  =\frac{1}{\sqrt{(|T|-1)|T|}}\sum_{m=1}^{|T|-1}
  \omega^{-j_0m}\omega^{j_0m}
  =
  \frac{|T|-1}{\sqrt{(|T|-1)|T|}}.
  \]
  Thus, the probability that the
  measurement gives $\vecb'\neq \vecb$ is
  $1-\frac{(|T|-1)^2}{(|T|-1)|T|}=1/|T|$.  Therefore, the probability
  the algorithm returns $\vecb' \neq \vecb$ is 
  $\frac{2}{|T|}(1-\frac{1}{|T|})$. 

  By Claim~\ref{claim:T_bound} the probability that a randomly chosen
  $\vecl$ satisfies $|T_{\vecl,p}| \leq poly(k)$ is at least $1/poly(k)$ when
  $k=\log N + c \log \log N$. Thus, the random collision problem can
  be solved by repeating the algorithm $poly(k)$ times.
\end{proof}

Now we consider the case where an arbitrary basis can be used within
each subspace spanned by solutions of a given subset sum instance
$\vecl$, $p$ as in Definition~\ref{def:newbasis}.  Let $\tUc$ be a
unitary that maps every state in $\tilde{\basisname}$ to quantum state
whose first qubit indicates if the state is in $\tilde{\basisname}^0$
or $\tilde{\basisname}^{\bot}$

\begin{corollary}\label{lem:tUc}
  If Algorithm~\ref{alg:bit} is run with $\tUc$ on input
  $\vecl,\vecb$, then it solves the collision problem with probability
  at least $\frac{1}{|T_{\vecl,p}|}(1-\frac{1}{|T_{\vecl,p}|})$, where
  $p =\vecl\cdot\vecb$.  In particular, if
  $k\leq \log N + c \log \log N$ then it solves the random collision
  problem in quantum polynomial time.
\end{corollary}

\begin{proof}
  Suppose $\tUc$ maps $|S_{\vecl,p}^0\rangle$ to a state
  $|0\rangle|\psi_{\vecl,p,0}\rangle$ and maps
  $|\tilde{S}_{\vecl,p}^m\rangle$ to
  $|1\rangle|\psi_{\vecl,p,m}\rangle$ for $m\in [1:|T|-1]$, where the
  set of vectors $\{|\psi_{\vecl,p,m}\rangle$ : $m\in [1:|T|-1]\}$ are
  an arbitrary orthonormal set of quantum states. The analysis is
  similar to the proof above, but we only consider the case when the
  state collapses to $m=0$. Specifically, after applying $\tUc$ to
  $|\vecb,\chi_\vecl\rangle$, the state is
\begin{equation}
    \tUc|\vecb,\chi_\vecl\rangle =
    \frac{1}{\sqrt{|T|}}|0\rangle|\psi_{\vecl,p,0}\rangle +
    \sqrt{\frac{|T|-1}{|T|}}(\sum_{m=1}^{|T|-1}c_m|1\rangle|\psi_{\vecl,p,m}\rangle).  
\end{equation}
 The probability the state collapses to $|0\rangle|\psi_{\vecl,p,0}\rangle$ after
 measuring the first qubit is $1/|T|$. After applying
 $\tUc^{\dagger}$ and measuring the state a vector $\vecb'\neq \vecb$
 is measured with probability $1-1/|T|$.  In total the probability of
 success is at least $\frac{|T|-1}{|T|^2}$.  For the choice of $k$
 given, this is at least $1/poly(k)$ with probability $1/poly(k)$ by
 Claim~\ref{claim:T_bound}.
\end{proof}

\bibliographystyle{plain}
\bibliography{fullpaper}

\appendix\label{appendix}
\section{Appendix}

\begin{claim} \label{claim:T_bound} Let $\vecb\in \mathbb{Z}_2^k$ be
  an arbitrary 
  fixed vector with $k=\log N + c\log \log N$ for some constant $c$.  Then
  over random choices of $\vecl \in \mathbb{Z}_{N}^k$, the probability
  that $|T_{\vecl,\vecb\cdot\vecl}| \leq poly(k)$ is at least
  $\frac{1}{poly(k)}$.
\end{claim}

\begin{proof}
  Fix $\vec{b}\in \mathbb{Z}_2^k$ and let $X_{\vecb'}$ be a random
  variable over $\vecl$ such that $X_{\vecb'}=1$ if
  $\vecb'\cdot\vecl=\vec{b}\cdot\vecl$ and  $X_{\vecb'}=0$ otherwise.
  Then $|T_{\vecl,\vecl\cdot\vecb}| = \sum_{\vecb'\in \Z_2^k}
  X_{\vecb'} = \sum_{\vecb'\in \Z_2^k\setminus\{\vecb\}}  X_{\vecb'}
  + 1$.  

For $\vecb'\neq \vecb$, the expected value of $X_{\vecb'}$ is $\Exp[X_{\vecb'}] 
    = \prob_\vecl (X_{\vecb'}=1)
    = \prob_\vecl (\vecl\cdot \vecb' = \vecl\cdot\vecb)
    =\frac{1}{N}.$
The last equality can be seen by choosing $i$ such that $b_i'=1$ and
$b_i=0$ without loss generality ($\vecb$ and $\vecb'$ can be swapped
if needed).  Then by fixing $l_j$ for $j\neq i$, and choosing
$l_i$ uniformly, $\vecb\cdot \vecl$ is fixed while $\vecb'\cdot \vecl$
is uniformly distributed. The variance of
$X_{\vec{b}'}$ is $\var(X_{\vec{b}'})=\frac{1}{N}-\frac{1}{N^2}.$

Therefore, the expected value of $|T_{\vecl,\vecl\cdot\vecb}|-1$ is
$\Exp[\sum_{\vecb'\in \mathbb{Z}_2^k\setminus\{\vec{b}\}}X_{\vecb'}]
=\sum_{\vecb'\in \mathbb{Z}_2^k\setminus\{\vec{b}\}
}\Exp[X_{\vec{b}'}]=\frac{2^k-1}{N}$,
and the variance of $|T_{\vecl,\vecl\cdot\vecb}|-1$ is
\begin{eqnarray}\label{eq:var}
  \var(\sum_{\vecb'\in\mathbb{Z}_2^k\setminus\{\vecb\}} X_{\vecb'})
  &=&\sum_{\vecb'\in\mathbb{Z}_2^k\setminus\{\vecb\}} \var(X_{\vec{b}'})
      +  \sum_{\vecb'\neq\vecb'', \vecb',\vecb''\in \mathbb{Z}_2^k\setminus \{\vecb\}} 
      \cov(X_{\vec{b}'},  X_{\vec{b}''})\nonumber\\ 
  &\leq& \frac{2^k-1}{N}+ \sum_{\vecb''\neq\vecb',\vecb',\vecb''\in
      \mathbb{Z}_2^k\setminus \{\vecb\}}\cov(X_{\vecb'},
      X_{\vecb''}).
\end{eqnarray}
This results in $\var(\sum_{\vecb'\in
  \mathbb{Z}_2^k\setminus\{\vec{b}\}}X_{\vecb'}) \leq \frac{2^k-1}{N}$
provided that the covariences are all zero, which we show below.
First we finish proving the claim by applying Chebyshev's inequality
to get
\[
	\prob(|T_{\vecl,\vecb\cdot\vecl}|\geq poly(k))\leq \frac{2^k-1}{N}
        \frac{1}{poly(k)}=\frac{1}{poly(k)},
\]
when $k\leq \log N + c \log \log N$.

In the following, we show that $X_{\vecb'}$ and $X_{\vecb''}$ are
independent when $\vecb$, $\vecb'$, and $\vecb''$ are all different
values, which implies $\cov(X_{\vec{b}'}, X_{\vec{b}''})=0$. To see
this let $1$ be a coordinate such that $b'_1=1$ and $b''_1=0$ without
loss of generality ($b'_j$ and $b''_j$ can be swapped). If $b_1=0$,
then
  \begin{eqnarray}
  	&&\prob_{\vecl}(X_{\vecb'}=1, X_{\vecb''}=1) \nonumber\\
	&=&\sum_{l_2,\dots,l_k = 0: }^{N-1} \prob_{l_1}(X_{\vecb'}=1,
            X_{\vecb''}=1| l_2,\dots, l_k)\cdot \prob(l_2,\dots, l_k) \nonumber\\ 
	&=&\frac{1}{N^{k-1}} \sum_{l_2,\dots,l_k = 0}^{N-1}
            \prob_{l_1}(X_{\vecb'}=1|l_2,\dots, l_k )\cdot
            \prob_{l_1}(X_{\vecb''}=1|l_2,\dots,
            l_k) \label{eq:indep_1_1}\\  
	&=&\frac{1}{N}\cdot \frac{1}{N^{k-1}}
            \sum_{l_2,\dots,l_k = 0}^{N-1}\prob_{l_1}(X_{\vecb''}=1|l_2,\dots, l_k)
            \nonumber\\ 
	&=&\frac{1}{N}\cdot \frac{1}{N^{k-1}} \cdot N^{k-2} = \frac{1}{N^2}. \label{eq:indep_1_2}
  \end{eqnarray}
  Equation~\ref{eq:indep_1_1} is true because $X_{b''}$ is fixed after
  fixing $l_2,\dots,l_k$.  For Equation~\ref{eq:indep_1_2} note that
  $\vecb$ and $\vecb''$ differ in at least one bit besides position
  $i=1$.  Therefore a $1/N$ fraction of the $N^{k-1}$ choices for
  $l_2,\dots,l_k$ satisfy $\vecl\cdot \vecb = \vecl\cdot \vecb''$.
  
  In the case where $b_1=1$ the properties of $X_{\vecb'}$ and
  $X_{\vecb''}$ are reversed:
  \begin{eqnarray}
  	&&\prob_{\vecl}(X_{\vecb'}=1, X_{\vecb''}=1) \nonumber\\
	&=&\sum_{l_2,\dots,l_k = 0}^{N-1}\prob_{l_1}(X_{\vecb'}=1,
            X_{\vecb''}=1| l_2,\dots, l_k )\cdot
            \prob(l_2,\dots, l_k) \nonumber\\ 
	&=&\frac{1}{N^{k-1}} \sum_{l_2,\dots,l_k = 0}^{N-1} 
            \prob_{l_1}(X_{\vecb'}=1|l_2,\dots, l_k )\cdot
            \prob_{l_1}(X_{\vecb''}=1|l_2,\dots,
            l_k)     \label{eq:indep_2_1}\\ 
	&=&\frac{1}{N} \frac{1}{N^{k-1}}
            \sum_{l_2,\dots,l_k = 0}^{N-1} \prob_{l_1} (X_{\vecb'}=1|l_2,\dots, l_k)  \nonumber\\ 
        &=&\frac{1}{N}\cdot \frac{1}{N^{k-1}} \cdot N^{k-2} =
            \frac{1}{N^2}. \label{eq:indep_2_2}
  \end{eqnarray}
  Equation~\ref{eq:indep_2_1} is true because $X_{\vecb'}$ is fixed to
  0 or 1 for all $l_1$.  Equation~\ref{eq:indep_2_2} is true because
  $\vecb$ and $\vecb'$ differ in at least one bit besides $i=1$.
  
Therefore, the covariance of $X_{\vecb'}$ and $X_{\vecb''}$ is $0$.  

\end{proof}

\end{document}